\documentclass[final,12pt]{colt2024}

\newcommand{\cV}{\mathcal{V}}
\newcommand{\cX}{\mathcal{X}}
\newcommand{\cY}{\mathcal{Y}}
\newcommand{\cP}{\mathcal{P}}

\newcommand{\bern}{\mathsf{Bern}}  

\definecolor{shcolor}{RGB}{27, 87, 14}

\newcommand{\shnote}[1]{\fcolorbox{outerlinecolor}{innerboxcolor}{
  \begin{minipage}{.9\textwidth}
    \textcolor{shcolor}{\bf SH Comment:} {#1}
\end{minipage}} \\
}

\renewcommand{\shnote}[1]{}

\newcommand{\thetav}{\theta(v)}
\newcommand{\thetavv}{\theta(v')}

\newcommand{\Pv}{\P_{v}}

\newcommand{\tiPv}{\tilde{\P}_{v}}

\newcommand{\unifV}{\mathsf{Uni}^{\cV}}

\newcommand{\acceptset}{A}

\newcommand{\distv}{P_{v}}
\newcommand{\distvv}{P_{v'}}
\newcommand{\tidistv}{\tilde{P}_{v}}
\newcommand{\tidistvv}{\tilde{P}_{v'}}

\newcommand{\conderr}{\mathsf{CondErr}}

\newcommand{\pfaccent}{\tilde}

\newcommand{\inlinemsg}[1]{~~\mbox{#1}}

\newcommand{\io}{\inlinemsg{infinitely often}}
\usepackage{macros}

\title{An information-theoretic lower bound in time-uniform estimation}

\coltauthor{%
	\Name{John Duchi} \Email{jduchi@stanford.edu}\\
	\addr Departments of Statistics and Electrical Engineering, Stanford University
	\AND
    \Name{Saminul Haque} \Email{saminulh@stanford.edu}\\
	\addr Department of Computer Science, Stanford University
}

\begin{document}
\maketitle



\begin{abstract}
  We present an information-theoretic lower bound for the problem of
  parameter estimation with time-uniform coverage guarantees.
  Via a new a reduction to sequential testing,
  we obtain stronger lower bounds
  that capture the hardness of the time-uniform setting.
  In the case of location model estimation, logistic regression, and
  exponential family models, our
  $\Omega(\sqrt{n^{-1}\log \log n})$ lower bound is sharp
  to within constant factors in typical settings.
\end{abstract}

\section{Introduction}
Let $X_1, X_2, \dots$ be independent samples from a common distribution
$P \in \cP$ over a domain $\cX$.
Given a parameter $\theta : \cP \rightarrow \Theta \subset \R$
and an error-bound function $t : \N \rightarrow \R_+$,
we say a sequence of estimators $\{\hat{\theta}_n\}_{n
  \in \N}$
is an \textit{$(\alpha, t(\cdot))$-estimator} if for all $P \in \cP$,
\begin{align}
  \P_{X_i \simiid P}\left(|\hat{\theta}_n(X_{1:n}) - \theta(P)| \leq t(n),
  \textrm{ for all } n \in \N\right) \geq 1-\alpha.
  \label{eq:t-unif-est}
\end{align}
We provide fundamental limits on $t(\cdot)$ that imply there exist
no $(\alpha, t(\cdot))$-estimators.

Our motivation comes from the problem of estimating time-uniform confidence
sequences~\citep{DarlingRo67,Lai76}, where the analyst must return a
confidence set $C_n$ at every $n$ such that $\P(\theta(P) \in C_n \textrm{
  for all } n) \geq 1-\alpha$.  \citet{DarlingRo67} introduced confidence
sequences as a means of allowing statistical inference without committing
\textit{a priori} to a fixed sample size \citep{Robbins70}.  Their
time-uniformity has made confidence sequences are a popular tool in
sequential data analysis, such as in clinical testing
\citep{JennisonTu89,Lai84}, A/B testing
\citep{JohariKoPeWa22,HowardRaMcSe21}, and bandit arm identification
\citep{JamiesonMaNoBu14,KaufmanCaGa16}.

Methods for producing confidence sequences often see the size of the
confidence intervals decay at a rate of $\Theta(\sqrt{n^{-1}\log \log n})$
\citep{DarlingRo67,HowardRaMcSe21,JamiesonMaNoBu14},
ignoring dependence on parameters other than $n$.
This stands in contrast to the fixed sample size setting, where the confidence
intervals instead are of size $\Theta(n^{-1/2})$; thus, these methods
incur a $\Theta(\sqrt{\log \log n})$ factor to achieve time-uniform coverage.
\citet{Farrell64} (and \citet{JamiesonMaNoBu14}, via a reduction to
\citeauthor{Farrell64}'s results) shows this cost is necessary
in exponential families.
Their lower bound technique proceeds in two steps:
first, they show an optimal estimator takes the form of a
sufficient statistic,
and second, they use the law-of-the-iterated-logarithm to show that this statistic
has fluctuations on the order of
$\Theta(\sqrt{n^{-1}\log \log n})$ infinitely often.
This argument crucially relies on the particular structure of the one-parameter
exponential family model, which allows one to argue that an optimal estimator
necessarily thresholds a sufficient statistic, making it unclear how, or even if,
the bounds generalize to other scenarios.
In this paper, we elicit the same
$\Theta(\sqrt{\log \log n})$
cost through information-theoretic techniques, which have the benefit of
extending to broader families of problems via now familiar
reductions~\citep[e.g.][Chapter 14]{Wainwright19}.


Our techniques rely on standard information-theoretic results, such as Le
Cam's two point method, the Bretagnolle-Huber inequality, and the reduction
from estimation to testing.  Our novel technique in proving this lower bound
is to reduce to a sequence of testing problems that grows iteratively harder.
By accumulating the errors from these tests, we obtain the extra
$\Theta(\sqrt{\log \log n})$ factor in the lower bound.


\paragraph{Notation}\label{sec:notation}
We use $\cV$ to denote the set of infinite binary sequences $\{0, 1\}^\infty$.
For $v \in \cV$ and $l, k \in \N$, we use $v_{k:l}$ to denote the
$(l-k+1)$-length binary string $(v_k,\dots, v_l)$.
If $l < k$, we understand $v_{k:l}$ to be the zero-length null sequence.
Given finite-length bit strings $b \in \{0, 1\}^k, b' \in \{0, 1\}^l$, we
use $b \oplus b'$ to denote the concatenation
$(b_1, \dots, b_k, b'_1, \dots, b'_l) \in \{0, 1\}^{k+l}$.
We use $\unifV$ to denote the uniform distribution over $\cV$.
Formally, if $V \sim \unifV$, then for any $k \in \N$ and $b \in \{0, 1\}^k$,
$\P(v_{1:k} = b) = 2^{-k}$.
Also for any $b \in \{0, 1\}^\star$, we use $\unifV_b$ to denote the distribution
$\unifV$ conditioned on the event that $b$ is a prefix.
We will use the property that for any
$b \in \{0, 1\}^\star$,
$\unifV_b = \tfrac{1}{2}\unifV_{b\oplus(0)} + \tfrac{1}{2}\unifV_{b\oplus(1)}$.

\section{Reduction from estimation to testing}
We adapt the classical approach to information-theoretic lower bounds
of reducing estimation to testing---showing that if one can solve the estimation
problem~\eqref{eq:t-unif-est}, then one can distinguish between many distributions.
Accordingly, let $\{\distv\}_{v \in \cV}$ be a family of probability distributions and
for each $v \in \cV$, let $\Pv$ be the probability space where
$X_1, X_2, \dots \simiid \distv$.
Under this setting, we define the following time-uniform test:
\begin{definition}[$(\alpha, \{n_k\}_{k \in \N})$-test]\label{def:test-seq}
  A sequence of $\{0, 1\}$-valued randomized tests
  $\{\hat{V}_k\}_{k \in \N}$ is an
  \emph{$(\alpha, \{n_k\})$-test} if for all $v \in \cV$,
  $\Pv(\hat{V}_k(X_{1:n_k}) = v_k \textrm{ for all } k \in \N) \geq 1-\alpha$.
\end{definition}

The testing problem will be to find an $(\alpha, \{n_k\})$-test given
$\{P_v\}$.
Conversely, we will show that if $\{n_k\}$ does not grow fast enough, then
no $(\alpha, \{n_k\})$-test can exist.
By reducing the original estimation problem~\eqref{eq:t-unif-est} to this
testing problem, the lower bound we find on $\{n_k\}$ will confer lower bounds
on $t(\cdot)$ for the original testing problem. 

\subsection{Reduction from estimation to testing}

We return to the estimation setting~\eqref{eq:t-unif-est}, where we have
distributions $\cP$ and a parameter
$\theta : \cP \rightarrow \Theta$.
For a sub-family $\{P_v\}_{v \in \cV} \subset \cP$,
say that $\{P_v\}$ has \textit{parameter separation} $\{\delta_k\}_{k \in \N}$
if for all $k \in \N$ and all $v, v' \in \cV$ such that the prefixes
$v_{1:k}\neq v'_{1:k}$, we have
$\delta_k \leq |\theta(\distv) - \theta(\distvv)|$.
This yields the following proposition which reduces estimation to testing.

\begin{proposition}\label{prop:reduction}
  Suppose $\{P_v\}_{v \in \cV} \subset \cP$ has parameter separation
  $\{\delta_k\}_{k \in \N}$ and $\hat{\theta}$ is an
  $(\alpha, t(\cdot))$-estimator.
  Then for $n_k = \inf\{n : t(n) < \delta_k/2\}$
  there exists an $(\alpha, \{n_k\})$-test.
\end{proposition}
\begin{proof}
  For ease of notation, we use $\thetav$ to denote $\theta(\distv)$
  and use $\hat{\theta}_n$ to refer to $\hat{\theta}_n(X_{1:n})$.
  We now construct a test sequence $\hat{V} = \{\hat{V}_k\}_{k \in \N}$.
  First define the projection function
  \begin{align*}
    \hat{v}(\vartheta) & = \argmin_{v \in \cV} |\vartheta - \thetav|,
  \end{align*}
  where we break potential ties arbitrarily, and then take
  $
    \hat{V}_k(X_{1:n_k}) = \hat{v}(\hat{\theta}_{n_k})_k.
  $

  To see that $\hat{V}$ is an $(\alpha, \{n_k\})$-test sequence,
  let $v \in \cV$ and consider the event
  $E = \cap_{n \in \N} \{|\hat{\theta}_n - \thetav| \leq t(n)\}$.
  By the time-uniform utility of $\hat{\theta}$,
  $\Pv(E) \geq 1-\alpha$, and so it suffices to show that on the event $E$,
  $\hat{V}_k = v_k$ for all $k \in \N$.
  It follows by definition of $\hat{v}$ and $n_k$ that for all $k \in \N$,
  \begin{align*}
    |\hat{\theta}_{n_k} - \theta(\hat{v}(\hat{\theta}_{n_k}))|
    \leq |\hat{\theta}_{n_k} - \thetav|
    \leq t(n_k)
    < \tfrac{1}{2}\delta_k.
  \end{align*}
  Therefore, $|\theta(\hat{v}(\hat{\theta}_{n_k})) - \thetav| < \delta_k$ by
  triangle inequality, which implies by the definition of $\delta_k$ that
  $v_{1:k} = \hat{v}(\hat{\theta}_{n_k})_{1:k}$ and, in particular,
  $v_k = \hat{v}(\hat{\theta}_{n_k})_k = \hat{V}_k$.
\end{proof}

\subsection{Hardness of the testing problem}

We say that $\{\distv\}$ has \textit{distribution closeness}
$\{\Delta_k\}_{k \in \N}$ if for all $k \in \N$ and for all $v, v' \in \cV$
such that $v_{1:k} = v'_{1:k}$, we have that $\dkl{\distv}{\distvv} \leq \Delta_k$.
Close distributions imply a lower bound on sequences $\{n_k\}$
for which $(\alpha, \{n_k\})$-tests exist.
\begin{theorem}
  \label{thm:testing-lb}
  Suppose $\{\distv\}_{v \in \cV}$ is a testing problem with distribution
  closeness $\{\Delta_k\}$.
  If there exists an $(\alpha, \{n_k\})$-test, then
  for infinitely many $k \in \N$,
      $n_k > (1-\alpha)\Delta_k^{-1}\log k$.
\end{theorem}



Combining Theorem~\ref{thm:testing-lb} with Proposition~\ref{prop:reduction}
immediately gives a lower bound for the time-uniform estimation problem.

\begin{corollary}
  \label{cor:est-lb}
  Let $\cP$ be a family of distributions and $\theta : \cP \rightarrow \R$ be
  a parameter function.
  Suppose there exists a sub-family $\{\distv\}_{v \in \cV} \subset \cP$ with
  parameter separation $\{\delta_k\}$ and distribution closeness $\{\Delta_k\}$.
  If there exists an $(\alpha, t(\cdot))$-estimator sequence, then
  for infinitely many $k \in \N$,
  \begin{align*}
    t\big((1-\alpha) \Delta_k^{-1} \log k\big) \geq \delta_k/2.
  \end{align*}
\end{corollary}


\begin{proof-of-theorem}[\ref{thm:testing-lb}]
  We first argue that we may take the tests to be deterministic without loss
  of generality.
  Consider the augmented problem setting of $\tilde{\cX} = \cX \times [0, 1]$ and
  $\tidistv = P_v \otimes \uniform([0, 1])$,
  so sample is supplemented with independent uniform randomness.
  Also define $\tiPv$ as the probability space corresponding to
  i.i.d. draws $\tilde{X}_1, \tilde{X}_2, \dots$ from $\tidistv$.

  We may express any randomized test $\hat{V}_k(X_{1:n_k})$ as a deterministic
  function $f_k$ of $X_{1:n_k}$ and some independent noise $U_k$ distributed
  according to $\uniform([0, 1])$.
  Then, we construct the deterministic test on $\tilde{\cX}^{n_k}$ by
  $\tilde{V}_k((X_i, U_i)_{1:n_k}) = f_k(X_{1:n_k}, U_{n_k})$.
  Applying this for all $k \in \N$ yields a deterministic test-sequence $\tilde{V}$
  on $\cX$ that has the same distribution as $\hat{V}$, and so if
  $\hat{V}$ is an $(\alpha, t(\cdot))$-test, then so is $\tilde{V}$.
  Finally, observe that $\dkls{\tidistv}{\tidistvv} = \dkl{P_v}{P_{v'}}$
  by chain rule of the KL-divergence,
  so the family $\{\tidistv\}_{v \in \cV}$ also has distribution closeness
  $\{\Delta_k\}$.
  Therefore, proving the lower bound for deterministic tests on
  $\{\tidistv\}_{v \in \cV}$ will imply the same lower bound for randomized
  tests on $\{P_v\}_{v \in \cV}$.
  
  We now define
  \begin{align*}
    \conderr(\tilde{V}, v, k) := \tiPv(\tilde{V}_k \neq v_k \mid \tilde{V}_{1:k-1} = v_{1:k-1}),
  \end{align*}
  which is the probability an $(\alpha, \{n_k\})$-test $\tilde{V}$ makes an error
  at test $k$ conditioned on the event that all previous tests were correct.
  The following lemma
  gives a lower bound on $\conderr(\tilde{V}, v, k)$ over
  uniform distributions:
  \begin{lemma}\label{lemma:indiv-cond-error-lb}
  Let $\pfaccent{V}$ be a deterministic $(\alpha, \{n_k\})$-test sequence
  for the testing problem $\{\tilde{\distv}\}$.
  For all $k \in \N$ and all $b \in \{0, 1\}^{k-1}$,
  \begin{align*}
    \int \conderr(\pfaccent{V}, v, k) d\unifV_{b}(v)
    \geq \frac{1}{4} \exp\left(-\frac{\Delta_k n_k}{1-\alpha}\right).
  \end{align*}
\end{lemma}

\begin{proof}
Because for each $j \in \N$,
$\pfaccent{V}_j$ is a deterministic function of $\tilde{X}_{1:n_j}$, for
each $b' \in \{0, 1\}^j$ we may define a set $\acceptset_{b'} \subset
\tilde{\cX}^{n_j}$ such that $\{\pfaccent{V}_{1:j} = b'\} = \{\pfaccent{X}_{1:n_j} \in
\acceptset_{b'}\}$.
Let $\cV_{b} = \{v \in \cV \mid v_{1:k-1} = b\}$ and note that $\cV_b$ is the
support of $\unifV_{b}$. 
For all $v \in \cV_b$,
we thus have the following equality of events in $\tiPv$:
\begin{equation}
  \{\pfaccent{V}_{1:k-1} = v_{1:k-1}\} = \{\pfaccent{V}_{1:k-1} = b\}
  = \{\tilde{X}_{1:n_{k-1}}\in \acceptset_{b}\}
  = \{\tilde{X}_{1:n_{k}}\in \acceptset_{b} \times \tilde{\cX}^{n_k - n_{k-1}}\}.
  \label{eqn:event-equality}
\end{equation}
Define $S := \acceptset_{b} \times
  \tilde{\cX}^{n_k - n_{k-1}} \subset \tilde{\cX}^{n_1}$
and notice that
$\acceptset_{b \oplus (0)}, \acceptset_{b \oplus (1)}$
are relative complements in $S$.
Then because $\unifV_{b} = \frac{1}{2}\unifV_{b\oplus(0)} +
  \frac{1}{2}\unifV_{b\oplus (1)}$, we get
\begin{align}
  \int \conderr(\pfaccent{V}, v, k) & d\unifV_{b}(v)
  = \int
  \tiPv(\pfaccent{V}_k \neq v_k \mid \pfaccent{V}_{k-1}=v_{k-1})
  d\unifV_{b}(v) \nonumber\\
  &= \frac{1}{2} \int \tiPv(\pfaccent{V}_k = 1 \mid
  \pfaccent{V}_{k-1}=v_{k-1}) d\unifV_{b\oplus(0)}(v) \nonumber
  \\&\quad+ \frac{1}{2} \int 
  \tiPv(\pfaccent{V}_k = 0 \mid \pfaccent{V}_{k-1}=v_{k-1})
  d\unifV_{b\oplus(1)}(v) \nonumber
  \\
  & \stackrel{(\star)}{=}
  \frac{1}{2} \int \tiPv(\tilde{X}_{1:n_k} \in \acceptset_{b \oplus (1)} 
    \mid \tilde{X}_{1:n_k} \in S) d\unifV_{b\oplus(0)}(v) \nonumber
  \\&\quad+ \frac{1}{2} \int \tiPv(\tilde{X}_{1:n_k} \in \acceptset_{b \oplus (0)}
    \mid \tilde{X}_{1:n_k} \in S) d\unifV_{b\oplus(1)}(v) \nonumber 
  \\
  \begin{split}
    &= \frac{1}{2} \int \tidistv^{n_k}(\acceptset_{b \oplus (1)} \mid S)
    d\unifV_{b\oplus(0)}(v) 
  \\&\quad
    + \frac{1}{2} \int \tidistv^{n_k}(\acceptset_{b \oplus (0)} \mid S)
      d\unifV_{b\oplus(1)}(v),   
  \end{split}
  \label{eq:conderr-decomp}
\end{align}
where equality~$(\star)$ follows from the identity~\eqref{eqn:event-equality}.

To proceed, we use Le Cam's method and the Bretagnolle-Huber inequality
below
(see, e.g.~\citet[Lemma 1]{Yu97} and
\citet[Lemma 2.6]{Tsybakov09} for proofs):
\begin{lemma}[Le Cam's method]\label{lemma:le-cam}
  Let $P, Q$ be distributions over a domain $\cX$.
  For any $S \subset \cX$,
  \begin{align*}
    P(S) + Q(S^c)
    \ge 1 - \sup_S (Q(S) - P(S))
    = 1 - \tvnorm{P - Q}.
  \end{align*}
\end{lemma}

\begin{lemma}[Bretagnolle-Huber]\label{lemma:bh}
  Let $P, Q$ be distributions over a domain $\cX$.
  Then
  \begin{align*}
    \tvnorm{P-Q} \leq 1 - \frac{1}{2}\exp(-\dkl{P}{Q}).
  \end{align*}
\end{lemma}

Applying Lemma~\ref{lemma:le-cam} to Eq.~\eqref{eq:conderr-decomp} implies
the lower bound
\begin{align*}
  &\int \conderr(\pfaccent{V}, v, k) d\unifV_{b}(v) \\
  &\qquad\geq \frac{1}{2}\left(1 -
  \tvnorm{\int
    \tidistv^{n_k}(\cdot \mid S)d\unifV_{b \oplus (0)}(v) - \int
    \tidistv^{n_k}(\cdot \mid S)d\unifV_{b \oplus (1)}(v)}
  \right).
\end{align*}
Then applying Lemma~\ref{lemma:bh} and using
the convexity of the KL-divergence, we obtain the lower bound
\begin{align}
  \int \conderr(\pfaccent{V}, v, k) d\unifV_{b}(v)
  &\geq \frac{1}{4}\exp\bigg(-\sup_{v, v' \in \cV_b}
  \dkl{\tidistv^{n_k}(\cdot \mid S)}{\tidistvv^{n_k}(\cdot \mid S)} \bigg).
  \label{eq:bh-lb}
\end{align}

To conclude the proof, we use the following lemma to bound the KL-divergence
between conditional distributions:
\begin{lemma}[KL-divergence of conditional distributions]\label{lemma:kl-cond-bd}
  Let $P, Q$ be two distributions over a domain $\cX$.
  Let $S$ be a measurable subset of $\cX$, and let $P_S, Q_S$ denote the
  conditional distributions of $P$ given $S$ and $Q$ given $S$, respectively.
  Then $\dkl{P_S}{Q_S} \leq \frac{1}{P(S)}\dkl{P}{Q}$.
\end{lemma}

\begin{proof}
  By the chain rule for KL-divergence,
  \begin{align*}
    \dkl{P}{Q}
    &= \dkl{\bern(P(S))}{\bern(Q(S))} + P(S)\dkl{P_S}{Q_S}
      + P(S^c)\dkl{P_{S^c}}{Q_{S^c}} \\
    &\geq P(S) \dkl{P_S}{Q_S}.
  \end{align*}
  Dividing by $P(S)$ proves the claim.
\end{proof}


By applying Lemma~\ref{lemma:kl-cond-bd} and the assumption
$\dkl{\tidistv}{\tidistvv} \leq \Delta_k$, we have
\begin{align*}
  \sup_{v, v' \in \cV_b}
  \dkl{\tidistv^{n_k}(\cdot \mid S)}{\tidistvv^{n_k}(\cdot \mid S)}
  &\leq \sup_{v, v' \in \cV_b} \frac{\dkl{\tidistv^{n_k}}{\tidistvv^{n_k}}}{\tidistv^{n_k}(S)}
  \leq \frac{\Delta_k n_k}{1-\alpha}.
\end{align*}
Substituting this into Eq.~\eqref{eq:bh-lb} proves
Lemma~\ref{lemma:indiv-cond-error-lb}.
\end{proof}

  Iteratively
  applying Lemma~\ref{lemma:indiv-cond-error-lb}, we obtain a lower bound on
  the cumulative conditional errors:

  \begin{lemma}
  \label{lemma:cond-err-lb}
  Let $\pfaccent{V}$ be a deterministic $(\alpha, \{n_k\})$-test sequence
  for the testing problem $\{\tilde{\distv}\}$.
  For all $\ell \in \N$, there exists a distribution $b \in \{0, 1\}^{\ell-1}$
  such that
  \begin{align*}
    \int \sum_{k=1}^\ell \conderr(\pfaccent{V}, v, k) d\unifV_b(v)
    \geq \frac{1}{4} \sum_{k=1}^\ell
    \exp\left(-\frac{\Delta_k n_k}{1-\alpha}\right).
  \end{align*}
\end{lemma}

  \begin{proof}
\newcommand{\prevb}{b'}
\newcommand{\currb}{b}
\newcommand{\bit}{i}
We proceed inductively.

\textbf{Base Case ($\ell=1$).}
This case follows immediately by applying Lemma~\ref{lemma:indiv-cond-error-lb}
with $k=1$.

\textbf{Inductive step.}
Let $\ell \geq 2$.
Assume the claim holds for $\ell-1$ and let $\prevb \in \{0, 1\}^{\ell-2}$
such that
\begin{align}
  \int \sum_{k=1}^{\ell-1} \conderr(\pfaccent{V}, v, k) d\unifV_{\prevb}(v)
  \geq \frac{1}{4} \sum_{k=1}^{\ell-1} 
    \exp\left(-\frac{\Delta_k n_k}{1-\alpha}\right).
  \label{eq:ind-ass}
\end{align}

It will suffice to find $\currb \in \{0, 1\}^{\ell-1}$ that
satisfies both
\begin{align}
  \int \sum_{k=1}^{\ell-1} \conderr(\pfaccent{V}, v, k) d\unifV_{\currb}(v)
  \geq \frac{1}{4} \sum_{k=1}^{\ell-1} \exp\left(-\frac{\Delta_k
    n_k}{1-\alpha}\right) \label{eq:ind-part-one}
\end{align}
and
\begin{align}
  \int \conderr(\pfaccent{V}, v, k)  d\unifV_{\currb}(v) \geq \frac{1}{4} \exp\left(-\frac{\Delta_\ell
    n_\ell}{1-\alpha}\right) \label{eq:ind-part-two}.
\end{align}

Because $\unifV_{\prevb} =
  \frac{1}{2}\unifV_{\prevb \oplus (0)} + \frac{1}{2}\unifV_{\prevb
    \oplus (1)}$,
the inductive hypothesis (Eq.~\eqref{eq:ind-ass}) implies
\begin{align}
  \max_{\bit \in \{0, 1\}}\int \sum_{k=1}^{\ell-1} \conderr(\pfaccent{V}, v, k)
    d\unifV_{\prevb \oplus (\bit)}(v)
  \geq \frac{1}{4} \sum_{k=1}^{\ell-1} 
    \exp\left(-\frac{\Delta_k n_k}{1-\alpha}\right).
  \label{eq:ind-part-one-using-hypothesis}
\end{align}
Let $\bit \in \{0, 1\}$ be any bit obtaining the maximum in the left-hand
side of Eq.~\eqref{eq:ind-part-one-using-hypothesis}.
By taking $\currb = \prevb \oplus (\bit)$, we ensure that
$\currb$ satisfies Eq.~\eqref{eq:ind-part-one}.
Applying Lemma~\ref{lemma:indiv-cond-error-lb} with $k=\ell$, we also have that
$\currb$ satisfies Eq.~\eqref{eq:ind-part-two}.
This completes the induction and proves the claim.
\end{proof}

  Conversely, the time-uniform conditions on $\tilde{V}$ yield an upper-bound on
  the sum of $\conderr$.

  \begin{lemma}
  \label{lemma:cond-err-ub}
  Let $\pfaccent{V}$ be an $(\alpha, \{n_k\})$-test sequence.
  For all $v \in \cV$,
  \begin{align}
    \sum_{k=1}^\infty \conderr(\pfaccent{V}, v, k) \leq \frac{\alpha}{1-\alpha}.
    \label{eq:ub-on-errs}
  \end{align}
\end{lemma}

  \begin{proof}
  By definition of $\hat{V}$, $\Pv(\hat{V}_{1:k}=v_{1:k}) \geq
    1-\alpha$ for all $v \in \cV$ and all $k \in \N$.
  Furthermore, for $v \in \cV$,
  \begin{align*}
    \Pv(\hat{V}_{1:k}=v_{1:k})
    &= \Pv(\hat{V}_{1:k-1}=v_{1:k-1}) \Pv(\hat{V}_k = v_k \mid
    \hat{V}_{1:k-1}=v_{1:k-1}) \\
    &= \Pv(\hat{V}_{1:k-1}=v_{1:k-1}) -
    \Pv(\hat{V}_{1:k-1}=v_{1:k-1})
    \Pv(\hat{V}_k \neq v_k \mid \hat{V}_{1:k-1}=v_{1:k-1}) \\
    &\leq \Pv(\hat{V}_{1:k-1}=v_{1:k-1}) - (1-\alpha)\conderr(\hat{V}, v, k).
  \end{align*}

  Iterating the above inequality thus yields
  $\Pv(\hat{V}_{1:k}=v_{1:k})
    \leq 1-(1-\alpha)\sum_{j=1}^k \conderr(\hat{V}, v, j)$.
  Because $\Pv(\hat{V}_{1:k}=v_{1:k}) \geq 1-\alpha$, we have that
  $\sum_{j=1}^k \conderr(\hat{V}, v, j) \leq
    \frac{\alpha}{1-\alpha}$.
  Taking $k \rightarrow \infty$ proves the result.
\end{proof}

  Combining Lemmas~\ref{lemma:cond-err-lb} and~\ref{lemma:cond-err-ub},
  we obtain
  \begin{align*}
    \frac{1}{4} \sum_{k=1}^\infty \exp\left(-\frac{\Delta_k n_k}{1-\alpha}\right)
    \leq \frac{\alpha}{1-\alpha}.
  \end{align*}
  Theorem~\ref{thm:testing-lb} therefore follows because the summability of the
  sequence implies that
  $\exp(-\frac{\Delta_k n_k}{1-\alpha})
  < \frac{1}{k}$, or
  $\log k < \frac{\Delta_k n_k}{1 - \alpha}$, infinitely often.
\end{proof-of-theorem}

\section{Applying the lower bound}

The abstract theorem and Corollary~\ref{cor:est-lb} combine to yield
a fairly straightforward recipe for proving lower bounds for time-uniform
estimation problems. For a parametric
model $\{P_\theta\}_{\theta \in \Theta}$, where $\Theta \subset \R$,
the basic idea proceeds in three steps:
\begin{enumerate}[1.]
\item Demonstrate the (typical) scaling that
  $\dkl{P_\theta}{P_{\theta'}} \le M (\theta - \theta')^2$ for
  some $M < \infty$
\item Develop a particular well-separated collection
  of parameters $\{\theta(v)\} \subset \Theta$ by, for
  $v \in \mc{V}$, writing a ternary expansion of
  $\theta(v) = 2 \sum_{i = 1}^\infty v_i 3^{-i}$, which guarantees
  that if $v, v'$ first disagree at position $k$,
  then $\dkl{P_\theta}{P_{\theta'}} \lesssim (3^{-k})^2$, while
  $|\theta(v) - \theta(v')| \gtrsim 3^{-k}$.
\item Apply Corollary~\ref{cor:est-lb}.
\end{enumerate}
This approach is familiar from classical ``two-point'' lower
bounds that rely on reducing estimation to testing~\citep[e.g.][]{Yu97,Duchi23}.
The quadratic scaling of the KL-divergence is indeed common:
in (essentially) any setting in which the Fisher Information $I_\theta$
of the
model $\{P_\theta\}$ exists, classical information geometry
considerations show that
$\dkl{P_\theta}{P_{\theta'}}
  = \frac{1}{2}I_\theta (\theta' - \theta)^2 + o(|\theta' - \theta|^2)$,
\citep[see, e.g.,][Exercise 11.7]{CoverTh06}.
Our lower bounds require the quadratic bound only locally, so
we expect them to apply in most estimation problems.

Let us formalize these steps. Let $\Theta$ be a convex subset of
$\R$ and $\theta_0 \in \Theta$, and let $r > 0$ be (a radius) such that
$[\theta_0, \theta_0 + r] \subset \Theta$.
To apply the lower bound, we wish to choose a parameter mapping 
$\theta : \cV \rightarrow \Theta$ to maximize the parameter separation
$\{\delta_k\}$ while minimizing the distribution closeness $\{\Delta_k\}$.
In our parametric setting, given a mapping $\theta : \cV \rightarrow \Theta$,
we can take
\begin{align*}
  \delta_k = \inf_{v_{1:k} \neq v'_{1:k}} |\theta(v) - \theta(v')|
  ~~~\qquad\textrm{and}\qquad~~~
  \Delta_k = \sup_{v_{1:k} = v'_{1:k}} M(\theta(v) - \theta(v'))^2.
\end{align*}
To that end,
we take a shifted Cantor set mapping
\begin{equation*}
  \theta(v) = \theta_0 + r \cdot 2 \sum_{i=1}^\infty v_i 3^{-i},
\end{equation*}
where by inspection $\theta(v) \in [\theta_0, \theta_0 + r] \subset \Theta$.
Upper and lower bounds on $|\thetav - \thetavv|$ for $v, v' \in \cV$ are
then immediate. Indeed, letting $v, v' \in \cV$ differ for the first time at
index $k$,
\begin{align*}
  |\thetav - \thetavv|
  = 2 r  \left| \sum_{i=k}^\infty (v_i - v_i') 3^{-i} \right|
  & = 2 r
  \cdot 3^{-k} \left| \sum_{i=0}^\infty (v_{i+k} - v_{i+k}') 3^{-i} \right| \\
  & \le 2 r  \cdot 3^{-k} \sum_{i = 0}^\infty 3^{-i}
  = 3^{1 - k} r  .
\end{align*}
So long as the KL-bound $\dkl{P_{\theta_0}}{P_{\theta_0 + \delta}}
\le M \delta^2$ holds for $\delta \le \frac{2}{3}r$, whenever
$v_{1:k} = v'_{1:k}$ we have $|\thetav - \thetavv| \le 3^{-k} r$, and
we may thus
choose distributional closeness
\begin{equation*}
  \Delta_k = M 9^{-k} r^2.
\end{equation*}
On the other hand, when $v, v'$ differ for the first
time in index $k$, then again by the triangle inequality,
\begin{align*}
  |\thetav - \thetavv|
  &\geq 2 r  \cdot 3^{-k}\bigg(
      |v_k - v_k'| - \bigg|\sum_{i=1}^\infty (v_{i+k} - v_{i+k}') 3^{-i}\bigg|
    \bigg)
  \geq 3^{-k} r
\end{align*}
because $v_k \neq v_k'$.  Therefore, if $v_{1:k} \neq v'_{1:k}$, then
$|\thetav - \thetavv| \geq 3^{-k} r$ and so we may take the parameter
separation $\delta_k = 3^{-k} r$.  Applying
Corollary~\ref{cor:est-lb} thus yields the following proposition.
\begin{proposition}
  \label{proposition:sams-burrito-recipe}
  Let $\Theta \subset \R$ be convex
  and the family $\{P_\theta\}_{\theta \in \Theta}$ be such that
  for some $r > 0$ and $\theta_0 \in \Theta$
  with $[\theta_0, \theta_0 + r] \in \Theta$,
  $\dkl{P_{\theta_0}}{P_{\theta_0 + \delta}}
  \le M \delta^2$ for $\delta \le r$. Then
  any $(\alpha, t(\cdot))$-estimator sequence satisfies
  \begin{equation*}
    t(n) \ge \sqrt{\frac{1 - \alpha}{8 M}}
    \cdot \sqrt{\frac{\log \log n}{n}} \io.
  \end{equation*}
\end{proposition}
\begin{proof}
  By Corollary~\ref{cor:est-lb}, if there exists an
  $(\alpha, t(\cdot))$-estimator, then for
  $N_k = (1-\alpha)\Delta_k^{-1} \log k$,
  we use the identification of the distributional closeness $\Delta_k
  = M \delta_k^2$ with the squared parameter separation to obtain
  \begin{align*}
    t(N_k) \ge \half \delta_k
    = \half \sqrt{\frac{\Delta_k}{M}}
    = \frac{1}{2}\sqrt{\frac{(1-\alpha) \log k}{M N_k}}
  \end{align*}
  for infinitely many $k \in \N$.
  Because $N_k = \frac{1 - \alpha}{M r^2} 9^k \log k$,
  for any sufficiently large $k \in \N$, we have both $N_k \leq 10^k$
  and 
  $\log \log N_k \leq \log k + \log \log 10 \leq 2 \log k$,
  or $\log k \ge  \half \log \log N_k$.
  Taking $n$ from among the indices $\{N_k\}_{k \in \N}$,
  \begin{align*}
    t(n) \geq \frac{1}{2}\sqrt{\frac{(1-\alpha) \log \log n}{2 M n}} \io
  \end{align*}
  as desired.
\end{proof}

While Proposition~\ref{proposition:sams-burrito-recipe}
demonstrates the $\log \log n$ penalty
of time-uniform estimation, a more standard technique demonstrates
the correct dependence on the confidence $\alpha$.

\shnote{Maybe because we drop the statistical rate in the main theorem, we can
remove this? this allows us to get back into the space constraints.}

\begin{proposition}
  \label{proposition:log-alpha-dependence}
  Let the conditions of
  Proposition~\ref{proposition:sams-burrito-recipe} hold.
  Then for $\alpha < \frac{1}{4}$,
  any $(\alpha, t(\cdot))$-estimator sequence satisfies
  \begin{align*}
    t(n) \geq \frac{1}{2}\sqrt{\frac{\log(1/4\alpha)}{Mn}}
    \inlinemsg{for~all~ $n \in \N$}.
  \end{align*}
\end{proposition}
\noindent
As the proof is a more-or-less standard two point lower bound,
we defer it to Section~\ref{sec:proof-log-alpha-dependence}.

\subsection{Location models}
\label{sec:loc-models}

We first demonstrate how the recipe above applies to obtain lower bounds on
(time-uniformly) estimating the parameter of a location model without
particular reliance on the structure of the family beyond mild integrability
conditions. Let $f$ be a density on $\R$, and consider the family
$\{P_\theta\}_{\theta \in \R}$ with densities of the form $x \mapsto f(x -
\theta)$.  Then under appropriate integrability conditions on the
derivatives of the density, the Fisher Information $I_\theta = \int
\frac{f'(x - \theta)^2}{f(x - \theta)} dx$ is constant~\citep{LehmannCa98},
and we also have the following quadratic expansion~\citep[e.g.][Section
  2.6]{Kullback68a}.
\begin{lemma}
  \label{lemma:densities-fisher}
  For sufficiently regular location families,
  $\dkl{P_\theta}{P_{\theta'}} = \half I_\theta (\theta - \theta')^2 + o(
  (\theta - \theta')^2)$.
\end{lemma}
\noindent

\begin{example}[Location families from regular densities]
  Let $f$ be a suitably regular density with Fisher Information $I_f = \int
  \frac{f'(x)^2}{f(x)} dx$. We may
  take $M = \half I_f$ in Proposition~\ref{proposition:sams-burrito-recipe},
  and so for any $(\alpha, t(\cdot))$-estimator,
  $t(n) \ge \half \sqrt{\frac{(1 - \alpha) \log \log
      n}{I_f \cdot n}}$ infinitely often.
\end{example}

We can be more explicit when it is easy to compute the KL-divergence.

\begin{example}[Gaussian location model]
  \label{ex:normal-loc}
  Consider $f(x) = \frac{1}{\sqrt{2\pi\sigma^2}}e^{-x^2/2\sigma^2}$.  Then
  for $P_\theta = \normal(\mu, \sigma^2)$, we have
  $\dkl{P_\theta}{P_{\theta'}} = \frac{1}{2 \sigma^2} (\theta-\theta')^2$,
  and so taking $M = \frac{1}{2\sigma^2}$
  Proposition~\ref{proposition:sams-burrito-recipe} yields that for any
  $(\alpha, t(\cdot))$-estimator, $t(n) \geq \frac{\sigma}{2}
  \sqrt{\frac{(1-\alpha)\log\log n}{n}}$ infinitely often.
\end{example}

\begin{example}[Cauchy location model]
  \label{ex:cauchy-loc}
  Consider the density $p_\theta(x) = 1 / (\pi \sigma [1 +
    (\frac{x-\theta}{\sigma})^2])$.  Then \citet[Theorem
    1]{ChyzakNi19}
  show that
  \begin{equation*}
    \dkl{P_\theta}{P_{\theta'}} = \log\left(1 +
    \frac{(\mu-\mu')^2}{4\sigma^2}\right) \leq
    \frac{(\mu-\mu')^2}{4\sigma^2},
  \end{equation*}
  and so in Proposition~\ref{proposition:sams-burrito-recipe}
  we take $M = \frac{1}{4\sigma^2}$ to obtain that
  for any $(\alpha, t(\cdot))$-estimator, $t(n) \geq \sigma
  \sqrt{\frac{(1-\alpha)\log\log n}{2 n}}$ infinitely often.
\end{example}

\subsection{Generalized linear models}

Exponential family models provide another clear demonstration of the lower
bound technique.  Consider two domains $\cX$ and $\cY$.  Let $Q$ be a
distribution over $\cX$, $\mu$ be a reference measure over $\cY$, and $T :
\cY \times \cX \rightarrow \R$ be a statistic of $y$ given $x$.  We define
the conditional exponential family model, or one-parameter generalized
linear model,
\begin{equation}
  \label{eqn:one-param-glm}
  p_\theta(y \mid x) =
  \exp\left(\theta T(y, x) - A(\theta \mid x)\right)
  ~~ \mbox{where} ~~
  A(\theta \mid  x) = \log \int_\cY \exp(\theta T(y \mid  x)) d\mu(y)
\end{equation}
and $p_\theta$ is a density with respect to $\mu$.
Then we define $P_\theta$ over $\cX \times \cY$ such $P_\theta(X \in \cdot) = Q$
and $P_\theta(Y \in \cdot \mid X = x)$ has density $p_\theta(y|x)$ with respect to
$\mu$.
The log-partition function $A(\theta \mid x)$ is analytic
on its open convex domain~\citep{Brown86}, and
\begin{align*}
  \dkl{P_\theta}{P_{\theta'}}
  & \stackrel{(i)}{=} \E_{X \sim Q}
  \left[\dkl{P_\theta(Y \in \cdot \mid X)}{P_{\theta'}(Y \in \cdot \mid X)}
    \right] \\
  & \stackrel{(ii)}{=} \E_{X \sim Q}
  \left[
    A(\theta' \mid X) - A(\theta \mid X) - A'(\theta \mid X)(\theta' - \theta)
    \right],
\end{align*}
where equality~$(i)$ follows from the chain rule and equality~$(ii)$
is an immediate
calculation for exponential family models. Recognizing the
Bregman divergence between $\theta$ and $\theta'$ according to
the (convex) log-partition function $A(\cdot \mid x)$
in this last display,
for all $\theta, \theta + \delta \in \dom A(\cdot \mid x)$ we obtain
\begin{equation*}
  \dkl{P_\theta(Y \in \cdot \mid x)}{P_{\theta + \delta}(Y \in \cdot \mid x)}
  = \half A''(\theta + t \mid x) \delta^2
\end{equation*}
for some $t \in [0, \delta]$. Because $A''(\theta \mid x) = \var_\theta(Y
\mid X = x)$ in the conditional model~\eqref{eqn:one-param-glm},
Proposition~\ref{proposition:sams-burrito-recipe} implies the following
corollary.
\begin{corollary}
  Assume the expected
  variance $\E[\var_\theta(Y \mid X)]$ is continuous and finite
  in a neighborhood of some $\theta_0$ in the one-parameter
  GLM~\eqref{eqn:one-param-glm}. Let
  $M = \E[\var_{\theta_0}(Y \mid X)]$.
  Then any $(\alpha, t(\cdot))$-estimator satisfies
  \begin{equation*}
    t(n) \ge \frac{1}{4} \sqrt{\frac{(1 - \alpha) \log \log n}{M n}}
    \io.
  \end{equation*}
\end{corollary}
\begin{proof}
  Let $M_\theta = \E[\var_\theta(Y \mid X)]$. Then
  by continuity and Taylor's theorem,
  $\dkl{P_{\theta_0}}{P_{\theta_0 + \delta}}
  = \half \E[A''(\theta_0 + t \delta \mid X)] \delta^2$
  for some $t \in [0, 1]$. Using continuity, there
  exists a neighborhood of $\theta_0$ such that
  $M_\theta \le 2 M$ for all $\theta$ in the neighborhood.
\end{proof}

We can make these calculations explicit for logistic regression:

\begin{example}[Logistic regression]\label{ex:log-reg}
  Take $\cX = \R$, $\cY = \{-1, 1\}$, $\mu$ to be the counting measure
  on $\cY$ and $T(y, x) = yx$.
  Then the log-partition $A(\theta \mid x) = \log(e^{-x\theta} + e^{x\theta})$
  satisfies
  \begin{align*}
    A'(\theta \mid x)
    = \frac{x(e^{x\theta}-e^{-x\theta})}{e^{-x\theta} + e^{x\theta}}
    ~~ \mbox{and} ~~
    A''(\theta \mid x)
    = x^2 - \frac{x^2 (e^{x\theta}-e^{-x\theta})^2 }{(e^{-x\theta} + e^{x\theta})^2}
    \leq x^2.
  \end{align*}
  Taking $M = \E[X^2]$,
  any $(\alpha, t(\cdot))$-estimator for
  $\theta$ satisfies
  $t(n) \geq \frac{1}{4}\sqrt{\frac{(1-\alpha)\log \log n}{\E[X^2] n}}$
  infinitely often.
\end{example}

\subsection{Models with nuisance parameters}
\label{sec:nuisances}

As we note above, the techniques we develop apply to any family with an
appropriate (Fisher) information measure, meaning that they transparently
apply to problems with nuisance parameters and in which we estimate some
functional of the distribution. To treat this in a fairly general case, we
consider semiparametric models, following~\citet[Chapter 25]{VanDerVaart98}.
For a distribution $P \in \cP$, we define the tangent space $\dot{\cP}_P$ to
be the set of functions $g \in L^2(P)$ such that for some neighborhood $O
\subset \R$ of $0$, there exists a one-dimensional parametric submodel
$\{P_h\}_{h \in O} \subset \cP$ satisfying $P_0 = P$ and
\begin{align*}
  \int (\sqrt{dP_h} - \sqrt{dP} - \tfrac{1}{2}hg\sqrt{dP})^2 = o(h^2).
\end{align*}
We say one-dimensional parametric submodels satisfying this are
\textit{quadratic mean differentiable (QMD) at $P$ with score
  $g$}~\citep[Chapter 7.2]{VanDerVaart98}; any such model admits the
quadratic approximation $\dkl{P}{P_h} = \frac{h^2}{2} \int g^2 dP + o(h^2)$,
and $g$ \emph{a fortiori} has mean 0 under $P$. We make the standard
assumption that the parameter $\theta$ is differentiable at $P$ with respect
to $\dot{\cP}_P$, meaning that it has an \emph{influence function}
$\dot{\theta}$ whose coordinate functions belong to the closure of the
linear span of $\dot{\cP}_P$ such that for all $\{P_h\} \subset \cP$ that
are QMD at $P$ with score $g$, we have
\begin{align}
  \theta(P_h) - \theta(P)
  &= h\int \dot{\theta}(x) g(x) dP(x) + o(h)
  ~~ \mbox{as~} h \to 0. \label{eq:derivative}
\end{align}

While these conditions are abstract, they arise naturally in parametric
models, and they also appear fairly generically in empirical risk
minimization problems (M-estimation), highlighting the power of the lower
bounds Theorem~\ref{thm:testing-lb} and
Proposition~\ref{proposition:sams-burrito-recipe} imply. We present two
examples before providing the general bound.
\begin{example}[Parametric models]
  \label{example:parametric-models}
  Let $\{P_{\beta}\}_{\beta \in B}$ be a parametric model family
  for some open $B \subset \R^d$, where
  we wish to estimate the parameter $\theta(P_\beta) = \varphi(\beta)$
  for a differentiable $\varphi : \R^d \to \R$.
  The simplest example consists of the coordinate functions
  $\varphi(\beta) = \beta_j$, the $j$th coordinate of $\beta$.  Let
  $\dot{\ell}_\beta = \nabla_\beta \log p_\beta$ be the gradient of the log
  likelihood and $I_\beta = P_\beta \dot{\ell}_\beta
  \dot{\ell}_\beta^T$ be the Fisher information matrix, where
  we assume that
  $I_\beta$ is positive definite.
  \Citet[Example 25.15]{VanDerVaart98} shows
  that
  $\dot{\theta}(x) = \nabla \varphi(\beta)^T I_\beta^{-1} \dot{\ell}_\beta(x)$
  is the influence function.
\end{example}

\begin{example}[Risk minimizers and M-estimators]
  \label{example:risk-minimizers}
  Let $\cP$ be the collection of distributions on a domain $\cX$ and
  consider a loss function $\ell : \Theta \times \cX \rightarrow \R_+$ where
  $\Theta \subset \R^d$ and $\ell$ is convex in $\theta$ for all
  $x \in \cX$. Define the risk minimizer $\theta : \cP \rightarrow
  \Theta$ by
  \begin{align*}
    \theta(P) = \argmin_{\theta \in \Theta}
    \left\{L(\theta) \defeq \E_P[\ell(\theta, X)]
    = \int \ell(\theta, x) dP(x) \right\}.
  \end{align*}
  Assuming sufficient smoothness of $L$ and $\ell$,
  the associated influence function is $\dot{\theta}(x) = - \nabla^2
  L(\theta(P))^{-1} \nabla_\theta \ell(\theta(P),
  x)$~\citep[e.g.][Prop.~1]{DuchiRu21}.
  Following the same approach as Example~\ref{example:parametric-models},
  an individual coordinate $\theta_j(P)$
  has influence
  \begin{equation*}
    \dot{\theta}_j(x) = -e_j^T \nabla^2 L(\theta(P))^{-1} \nabla_\theta
    \ell(\theta(P), x).
  \end{equation*}
\end{example}

Regardless, we have the following lower bound
whenever an influence function exists.

\begin{proposition}
  \label{prop:semiparam-lb}
  Let $\theta: \cP \rightarrow \R$ have influence $\dot{\theta}$
  at $\dot{\mc{P}}_P$.
  Then any $(\alpha, t(\cdot))$-estimator of $\theta$ satisfies
  \begin{align*}
    t(n) \geq \frac{1}{4} \sqrt{(1-\alpha) \int \dot{\theta}(x)^2 dP(x)}
      \cdot \sqrt{\frac{\log \log n}{n}} \io.
  \end{align*}
\end{proposition}
\begin{proof}
  Because minimax estimation over all of $\cP$ is at least as hard as
  minimax estimation over any one-dimensional subfamily, we will demonstrate
  minimax lower bounds over a QMD subfamily with score $g : \mc{X} \to \R$
  and then choose an appropriate $g$.

  Let $\{P_h\}_{h \in H}$ be a QMD subfamily with score $g$,
  where $H$ is a neighborhood of 0, and we assume
  that $|\int \dot{\theta}(x) g(x) dP(x)| > 0$ in
  the first-order expansion~\eqref{eq:derivative}
  (as we will choose $g$ later, this is no loss of generality).
  We reduce
  estimation of $\theta$ over $\{P_h\}$ to estimation of $h \in \R$
  itself.
  For any $\delta > 0$, by
  restricting $H$ to a small enough neighborhood of 0,
  the derivative~\eqref{eq:derivative} implies that for
  any $h, h' \in H$ we have
  \begin{equation*}
    \theta(P_h) - \theta(P_{h'})
    = \theta(P_h) - \theta(P) + \theta(P) - \theta(P_{h'})
    = (h - h') \int \dot{\theta}(x) g(x) dP(x)
    \pm \delta.
  \end{equation*}
  Now, given an estimator $\hat{\theta}_n$, define $\hat{h}_n =
  \argmin_{h' \in H} |\hat{\theta}_n-\theta(P_{h'})|$
  (or a value arbitrarily close to minimizing $|\hat{\theta}_n
  -\theta(P_{h'})|$), so that
  \begin{align*}
    |\theta(P_{\hat{h}_n}) - \theta(P_h)|
    \leq |\theta(P_{\hat{h}_n}) - \hat{\theta}_n| + |\hat{\theta}_n - \theta(P_h)|
    \leq 2|\hat{\theta}_n - \theta(P_h)|.
  \end{align*}
  Rearranging,
  $|\hat{h}_n - h|
  \leq \frac{2}{|\int \dot{\theta}(x) g(x) dP(x)| - \delta}
  |\hat{\theta}_n - \theta(P_h)|$, and
  we have shown that any $(\alpha,
  t(\cdot))$-estimator for $\theta$ on $\{P_h\}$ yields an
  $(\alpha, \frac{2}{|\int \dot{\theta}(x) g(x) dP(x)| - \delta}
  t(\cdot))$-estimator  $\hat{h}_n$ for $h$ on $\{P_h\}$.

  Because $\dkl{P}{P_h} =
  \frac{h^2}{2} \int g^2 dP + o(h^2)$,
  Proposition~\ref{proposition:sams-burrito-recipe} implies
  that for any $M >
  \frac{1}{2} \int g^2 dP$, any $(\alpha, \tilde{t}(\cdot))$-estimator for
  $h$ satisfies
  \begin{align*}
    \tilde{t}(n) \geq \sqrt{\frac{1-\alpha}{8M}} \cdot \sqrt{\frac{\log \log n}{n}}
    \io.
  \end{align*}
  Putting these two observations together and taking $\delta \searrow 0$ and
  $M \searrow \frac{1}{2} \int g^2 dP$, we conclude that if
  $\hat{\theta}_n$ is an $(\alpha, t(\cdot))$-estimator for $\theta$ on $\{P_h\}$,
  then
  \begin{align*}
    t(n) \geq \frac{|\int \dot{\theta}(x) g(x) dP(x)|}{2}
      \cdot \sqrt{\frac{1-\alpha}{4\int g^2 dP}} \cdot \sqrt{\frac{\log \log n}{n}}
    \io.
  \end{align*}
  This lower bound over $\{P_h\}$ trivially provides a
  minimax lower bound for $\cP$, so
  choosing $g = \dot{\theta}$
  yields the proposition.
\end{proof}

Returning to Examples~\ref{example:parametric-models}
and~\ref{example:risk-minimizers}, Proposition~\ref{prop:semiparam-lb}
provides immediate lower bounds.
In Example~\ref{example:parametric-models},
we have influence
$\dot{\theta} = \nabla \varphi(\beta)^T I_\beta^{-1} \dot{\ell}_\beta$,
yielding
\begin{align*}
  \E_{P_\beta}[\dot{\theta}(X)^2] =
  \nabla \varphi(\beta)^T \E_{P_\beta}\left[I_\beta^{-1}
    \dot{\ell}_\beta \dot{\ell}_\beta^T I_\beta^{-1}\right]
  \nabla \varphi(\beta)
  = \nabla \varphi(\beta)^T I_\beta^{-1} \nabla \varphi(\beta),
\end{align*}
where $I_\beta$ is the Fisher information for the model $P_\beta$.
Applying Proposition~\ref{prop:semiparam-lb} and taking the supremum over
$\beta \in B$, we obtain that for all $(\alpha, t(\cdot))$-estimators,
\begin{align*}
  t(n)
  \geq \frac{1}{4} \sqrt{(1-\alpha)
    \sup_{\beta \in B}
    \nabla \varphi(\beta)^T I_\beta^{-1} \nabla \varphi(\beta)}
  \cdot \sqrt{\frac{\log \log n}{n}} \io.
\end{align*}
Similarly, in Example~\ref{example:risk-minimizers}, if we define the
asymptotic efficient covariance
\begin{equation*}
  \Sigma_P = \nabla^2 L(\theta(P))^{-1}
  \cov(\nabla \ell(\theta(P), X)) \nabla^2 L(\theta(P))^{-1},
\end{equation*}
any $(\alpha, t(\cdot))$-estimator of $\varphi(\theta(P))$ for a
differentiable $\varphi : \R^d \to \R$ must satisfy
\begin{align*}
  t(n)
  \geq \frac{1}{4} \sqrt{(1-\alpha)
    \nabla \varphi(\theta(P))^T \Sigma_P
    \nabla \varphi(\theta(P))}
  \cdot \sqrt{\frac{\log \log n}{n}} \io.
\end{align*}

\section{Attaining the lower bound}

\newcommand{\tu}[1]{\hat{#1}^{\textup{tu}}}

Via a suitable doubling strategy, estimators achieving high-probability
convergence guarantees for fixed sample sizes essentially immediately
provide time-uniform guarantees with an additional $\log \log n$ penalty,
matching our lower bounds.  Readers should not take this section as
advocating particular estimators---substantial
work~\citep[e.g.][]{HowardRaMcSe20} on time-uniform sequences
develops stronger estimators---but as a relatively simple proof
of concept (see also~\citet[Thm.~3.1]{KirichenkoGr21}).
Given a sequence of estimators $\hat{\theta}_n$,
consider the \emph{time uniform} estimator-sequence
\begin{align}
  \tu{\theta}_n
  & \defeq
  \hat{\theta}_{2^{\floor{\log_2 n}}}
  =
  \begin{cases}
    \hat{\theta}_n(X_{1:n}) & ~ \mbox{if} ~ n = 2^k,
    \textrm{ for some } k \in \N \\
    \tu{\theta}_{n-1}      & \textrm{otherwise},
  \end{cases}
  \label{eq:upper-bound}
\end{align}
so that $\tu{\theta}$ is $\hat{\theta}_n$ evaluated at the
most recent power of two sample size.
So long as $\hat{\theta}$ is sufficiently accurate,
$\tu{\theta}_n$ is a time-uniform estimator. In stating the
proposition, we say that $\hat{\theta}$ has
\emph{high-probability deviation bound $F$} for $F : \R_+ \times \N \to \R_+$
if 
\begin{equation*}
  \P\left(|\hat{\theta}_n(X_{1:n}) - \theta(P)| \leq
  F(\log(1/\alpha), n)\right)
  \geq 1-\alpha
  ~~ \mbox{for~all~}
  \alpha \in [0,1], ~ n \in \N,
\end{equation*}
where $F$ is increasing in its first argument and decreasing in its
second. In typical cases, such as sub-Gaussian mean estimation, we expect
$F(t, n) = C \sqrt{t / n}$ for some $C$.
\begin{proposition}
  \label{prop:upper-bound}
  Let $\hat{\theta}_n$ have high probability deviation bound $F$.
  Then $\{\tu{\theta}_n\}$ is an $(\alpha, t(\cdot))$-estimator for
  \begin{equation*}
    t(n)
    = F\big(2\log \log_2 n + \log(1 / \alpha) + 1/2, n/2\big).
  \end{equation*}
\end{proposition}
\begin{proof}
  Let $E$ denote the event that for all $k \in \N$, $|\hat{\theta}_{2^k} -
  \theta(P)| \leq F(\log(\pi^2 k^2/6\alpha), 2^k)$.  By a union bound,
  $\P(E^c) \le \frac{6}{\pi^2} \sum_{k
    = 1}^\infty \frac{\alpha}{k^2} = \alpha$ by the Basel problem,
  so $E$ occurs with probability at least $1-\alpha$.
  Now for $n \in \N$, let $k = k_n = \floor{\log_2 n}$
  be the unique natural number such that
  $2^k \leq n < 2^{k+1}$.  By definition, $\hat{\theta}_n =
  \hat{\theta}_{2^k}$ and so on $E$,
  \begin{align*}
    |\tu{\theta}_n - \theta(P)|
    \leq F\left(\log\frac{\pi^2 k^2}{6\alpha}, 2^k\right)
    \inlinemsg{for all $n$}.
  \end{align*}
  Now use
  that $k \le \log_2 n$
  implies $\log k^2 = 2 \log k \le 2 \log \log_2 n$,
  $\log \frac{\pi^2}{6} < \half$, and that
  $F$ is decreasing in its second argument and
  $2^k > n/2$.
\end{proof}

As a simple example, consider estimating a sub-Gaussian mean,
where we recall that $X$ is $\sigma^2$-sub-Gaussian if
$\E[e^{\lambda(X - \E[X])}] \le \exp(\frac{\lambda^2 \sigma^2}{2})$
for all $\lambda \in \R$.

\begin{example}[Sub-Gaussian mean estimation]
  Let $X_i \simiid P$ be $\sigma^2$-sub-Gaussian with mean $\theta =
  \theta(P)$.
  The sample mean $\hat{\theta}_{n} = \frac{1}{n}\sum_{i=1}^n X_i$
  satisfies $\P(|\hat{\theta}_n - \theta| \ge t) \le 2 \exp(-\frac{n t^2}{2
    \sigma^2})$, that is,
  $\hat{\theta}_n$ has high probability deviation
  bound $F(\log\frac{1}{\alpha}, n)
  = \sqrt{\frac{2 \sigma^2 \log (2/\alpha)}{n}}$.
  The
  time-uniform estimator~\eqref{eq:upper-bound}
  is thus an
  $(\alpha, t(\cdot))$-estimator for
  \begin{align*}
    t(n)&= O(1) \cdot
    \sigma\sqrt{\frac{\log \log n + \log(1 / \alpha)}{n}}
  \end{align*}
  by  Proposition~\ref{prop:upper-bound},
  matching the lower bound in
  Example~\ref{ex:normal-loc} to within constant factors;
  Proposition~\ref{proposition:log-alpha-dependence} shows its
  sharpness generally.
\end{example}

\section{Proof of Proposition~\ref{proposition:log-alpha-dependence}}
\label{sec:proof-log-alpha-dependence}

For $\theta \in \Theta$, let $\P_\theta$ represent the probability space
corresponding to $X_1, X_2, \ldots \simiid P_\theta$.
Considering testing between two points
$\theta_0$ and $\theta_\delta = \theta_0 + \delta \in \Theta$, 
where $\delta > 0$ is a quantity we choose later.
For any test $T_n$ based on a sample of size $n$, Le Cam's lemma
(Lemma~\ref{lemma:le-cam}) and the Bretagnolle-Huber inequality
(Lemma~\ref{lemma:bh}) give that
\begin{align*}
  \P_{\theta_0}(T_n \neq \theta_0) + \P_{\theta_\delta}(T_n \neq \theta_\delta)
  &\geq \frac{1}{2}\exp(-\dkl{P_{\theta_0}^n}{P_{\theta_\delta}^n})
  \geq \frac{1}{2}\exp(-nM\delta^2),
\end{align*}
where we use the assumption that
$\dkl{P_{\theta_0}}{P_{\theta_0 + \delta}}
\leq M \delta^2$.
One such test is 
\begin{align*}
  T_n
  &\defeq \begin{cases}
    \theta_0 & \textrm{if } |\theta_0 - \hat{\theta}_n| \leq |\hat{\theta}_n - \theta_\delta| \\
    \theta_\delta & \textrm{otherwise}.
  \end{cases}
\end{align*}
By the triangle inequality $|\hat{\theta}_n - \theta_0| \le
\frac{\delta}{2}$ implies $|\hat{\theta}_0 - \theta_n|
\le |\hat{\theta}_n - \theta_\delta|$ and $T_n = \theta_0$, and similarly,
$|\hat{\theta}_n - \theta_\delta| < \delta/2$ implies
$T_n = \theta_\delta$. Thus
\begin{align*}
  \P_{\theta_0}(|\hat{\theta}_n - \theta_0| > \delta/2)
  +\P_{\theta_\delta}(|\hat{\theta}_n - \theta_\delta| \geq \delta/2)
  &
  \ge \P_{\theta_0}(T_n \neq \theta_0)
  + \P_{\theta_\delta}(T_n \neq \theta_\delta)
  \geq \frac{1}{2}\exp(-nM\delta^2),
\end{align*}
and so
\begin{align*}
  \sup_{\theta \in \Theta} \P_\theta(|\hat{\theta}_n-\theta| \geq \delta/2)
  \geq \max_{\theta \in \{\theta_0, \theta_\delta\}}
  \P_\theta(|\hat{\theta}_n-\theta| \geq \delta/2)
  \geq \frac{1}{4}\exp(-nM\delta^2).
\end{align*}
By choosing
$\delta = \sqrt{\log(1/4\alpha')/Mn}$ for any $\alpha' > \alpha$, we have
\begin{align*}
  \inf_{\theta \in \Theta}
  \P_\theta\left(|\hat{\theta}_n-\theta| \leq \frac{1}{2}\sqrt{\frac{\log(1/4\alpha')}{Mn}}\right)
  \leq 1-\alpha' < 1-\alpha.
\end{align*}
On the other hand, the time-uniform condition gives that for any $n$,
\begin{align*}
  \inf_{\theta \in \Theta}
  \P_{\theta}(|\hat{\theta}_n - \theta| \leq t(n))
  \geq \inf_{\theta \in \Theta}
  \P_\theta(|\hat{\theta}_m - \theta| \leq t(m), \text{ for all } m \in \N)
  \geq 1-\alpha.
\end{align*}
This implies that $t(n) \geq \tfrac{1}{2}\sqrt{\log(1/4\alpha')/Mn}$.
Taking $\alpha' \searrow \alpha$ thus proves the claim.

\section*{Acknowledgments}

Thank you to the anonymous reviewers, who provided the simple proof of
Lemma~\ref{lemma:indiv-cond-error-lb} and encouraged us to include the
material in Section~\ref{sec:nuisances}.  This work was partially
supported by the Office of Naval Research grant N00014-22-1-2669 and
the Stanford DAWN Consortium.

\bibliography{bib}

\begin{thebibliography}{23}
\providecommand{\natexlab}[1]{#1}
\providecommand{\url}[1]{\texttt{#1}}
\expandafter\ifx\csname urlstyle\endcsname\relax
  \providecommand{\doi}[1]{doi: #1}\else
  \providecommand{\doi}{doi: \begingroup \urlstyle{rm}\Url}\fi

\bibitem[Brown(1986)]{Brown86}
Lawrence~D. Brown.
\newblock \emph{Fundamentals of Statistical Exponential Families}.
\newblock Institute of Mathematical Statistics, Hayward, California, 1986.

\bibitem[Chyzak and Nielsen(2019)]{ChyzakNi19}
Fr\'{e}d\'{e}ric Chyzak and Frank Nielsen.
\newblock A closed-form formula for the {K}ullback-{L}eibler divergence between {C}auchy distributions.
\newblock \emph{arXiv:1905.10965 [cs.IT]}, 2019.

\bibitem[Cover and Thomas(2006)]{CoverTh06}
Thomas~M. Cover and Joy~A. Thomas.
\newblock \emph{Elements of Information Theory, Second Edition}.
\newblock Wiley, 2006.

\bibitem[Darling and Robbins(1967)]{DarlingRo67}
Donald~A Darling and Herbert Robbins.
\newblock Confidence sequences for mean, variance, and median.
\newblock \emph{Proceedings of the National Academy of Sciences}, 58\penalty0 (1):\penalty0 66, 1967.

\bibitem[Duchi(2023)]{Duchi23}
John~C. Duchi.
\newblock Stats311/{EE}377: Information theory and statistics.
\newblock Course at Stanford University, Fall 2023.
\newblock URL \url{http://web.stanford.edu/class/stats311}.

\bibitem[Duchi and Ruan(2021)]{DuchiRu21}
John~C. Duchi and Feng Ruan.
\newblock Asymptotic optimality in stochastic optimization.
\newblock \emph{Annals of Statistics}, 49\penalty0 (1):\penalty0 21--48, 2021.

\bibitem[Farrell(1964)]{Farrell64}
Roger~H. Farrell.
\newblock Asymptotic behavior of expected sample size in certain one sided tests.
\newblock \emph{Annals of Mathematical Statistics}, 35\penalty0 (1):\penalty0 36--72, 1964.

\bibitem[Howard et~al.(2020)Howard, Ramdas, McAuliffe, and Sekhon]{HowardRaMcSe20}
Steven~R. Howard, Aaditya Ramdas, Jon McAuliffe, and Jasjeet Sekhon.
\newblock Time-uniform {C}hernoff bounds via nonnegative supermartingales.
\newblock \emph{Probability Surveys}, 17:\penalty0 257--317, 2020.

\bibitem[Howard et~al.(2021)Howard, Ramdas, McAuliffe, and Sekhon]{HowardRaMcSe21}
Steven~R. Howard, Aaditya Ramdas, Jon McAuliffe, and Jasjeet Sekhon.
\newblock Time-uniform, nonparametric, nonasymptotic confidence sequences.
\newblock \emph{Annals of Statistics}, 49\penalty0 (2):\penalty0 1055--1080, 2021.

\bibitem[Jamieson et~al.(2014)Jamieson, Malloy, Nowak, and Bubeck]{JamiesonMaNoBu14}
Kevin Jamieson, Matthew Malloy, Robert Nowak, and Sebastien Bubeck.
\newblock lil'{UCB}: An optimal exploration algorithm for multi-armed bandits.
\newblock In \emph{Proceedings of the Twenty Seventh Annual Conference on Computational Learning Theory}, pages 423--439. PMLR, 2014.

\bibitem[Jennison and Turnbull(1989)]{JennisonTu89}
Christopher Jennison and Bruce~W Turnbull.
\newblock Interim analyses: the repeated confidence interval approach.
\newblock \emph{Journal of the Royal Statistical Society, Series B}, 51\penalty0 (3):\penalty0 305--334, 1989.

\bibitem[Johari et~al.(2022)Johari, Koomen, Pekelis, and Walsh]{JohariKoPeWa22}
Ramesh Johari, Pete Koomen, Leo Pekelis, and David~J. Walsh.
\newblock Always valid inference: Continuous monitoring of {A/B} tests.
\newblock \emph{Operations Research}, pages 1806--1821, 2022.

\bibitem[Kaufmann et~al.(2016)Kaufmann, Capp\'{e}, and Garivier]{KaufmanCaGa16}
Emilie Kaufmann, Olivier Capp\'{e}, and Aur\'{e}lien Garivier.
\newblock On the complexity of best arm identification in multi-armed bandit models.
\newblock \emph{Journal of Machine Learning Research}, 17\penalty0 (1):\penalty0 1--42, 2016.

\bibitem[Kirichenko and Grunwald(2021)]{KirichenkoGr21}
Alisa Kirichenko and Peter Grunwald.
\newblock Minimax rates without the fixed sample size assumption.
\newblock \emph{arXiv:2006.11170v2 [math.ST]}, 2021.

\bibitem[Kullback(1968)]{Kullback68a}
Solomon Kullback.
\newblock \emph{Information Theory and Statistics}.
\newblock Dover Publications, 1968.

\bibitem[Lai(1976)]{Lai76}
Tze-Leung Lai.
\newblock On confidence sequences.
\newblock \emph{Annals of Statistics}, 4\penalty0 (2):\penalty0 265--280, 1976.

\bibitem[Lai(1984)]{Lai84}
Tze-Leung Lai.
\newblock Incorporating scientific, ethical and economic considerations into the design of clinical trials in the pharmaceutical industry: a sequential approach.
\newblock \emph{Communications in Statistics -- Theory and Methods}, 13\penalty0 (19):\penalty0 2355--2368, 1984.

\bibitem[Lehmann and Casella(1998)]{LehmannCa98}
Erich~L. Lehmann and George Casella.
\newblock \emph{Theory of Point Estimation, Second Edition}.
\newblock Springer, 1998.

\bibitem[Robbins(1970)]{Robbins70}
Herbert Robbins.
\newblock Statistical methods related to the law of the iterated logarithm.
\newblock \emph{Annals of Mathematical Statistics}, 41\penalty0 (5):\penalty0 1397--1409, 1970.

\bibitem[Tsybakov(2009)]{Tsybakov09}
Alexandre~B. Tsybakov.
\newblock \emph{Introduction to Nonparametric Estimation}.
\newblock Springer, 2009.

\bibitem[van~der Vaart(1998)]{VanDerVaart98}
Aad~W. van~der Vaart.
\newblock \emph{Asymptotic Statistics}.
\newblock Cambridge Series in Statistical and Probabilistic Mathematics. Cambridge University Press, 1998.

\bibitem[Wainwright(2019)]{Wainwright19}
Martin~J. Wainwright.
\newblock \emph{High-Dimensional Statistics: A Non-Asymptotic Viewpoint}.
\newblock Cambridge University Press, 2019.

\bibitem[Yu(1997)]{Yu97}
Bin Yu.
\newblock Assouad, {F}ano, and {L}e {C}am.
\newblock In \emph{Festschrift for Lucien Le Cam}, pages 423--435. Springer-Verlag, 1997.

\end{thebibliography}


\end{document}